\documentclass[]{article}
\usepackage{amsmath}
\usepackage{amssymb}
\usepackage{amsthm}
\usepackage{amsbsy}
\usepackage{amscd}
\usepackage[all]{xy}
\usepackage{makeidx}
\usepackage{wasysym}
\usepackage{yfonts}
\usepackage{mathtools}
\usepackage{stmaryrd}
\usepackage{tikz}
\usepackage{graphicx}
\usepackage{aurical}
\usepackage{xparse}
\usepackage{dsfont}
\usepackage{tablefootnote}
\usepackage[format=plain,margin=1cm,
labelfont={bf,it},
textfont=it]{caption}
\usepackage{array} 
\usepackage{tabularx} 

\usepackage[square,numbers]{natbib}
\usepackage[T1]{fontenc}
\usepackage{xparse}
\usepackage[format=plain,margin=1cm,
labelfont={bf,it},
textfont=it]{caption}
\newcommand{\hooklongrightarrow}{\lhook\joinrel\longrightarrow}

\newtheorem{thrm}{Theorem}
\numberwithin{thrm}{section}
\newtheorem{prop}{Proposition}
\numberwithin{prop}{section}
\newtheorem{lemma}{Lemma}
\numberwithin{lemma}{section}
\newtheorem{defn}{Definition}
\numberwithin{defn}{section}
\newtheorem{expl}{Example}
\numberwithin{expl}{section}
\theoremstyle{remark}
\newtheorem{rmk}{Remark}
\numberwithin{rmk}{section}
\numberwithin{equation}{section}

\NewDocumentCommand{\dbkt}{sO{}m}{%
	\IfBooleanTF{#1}
	{\dbktext{#3}}
	{\dbktx[#2]{#3}}%
}

\NewDocumentCommand{\dbktext}{m}{%
	\sbox0{%
		\mathsurround=0pt 
		$\left\{\vphantom{#1}\right.\kern-\nulldelimiterspace$%
	}%
	\sbox2{\{}%
	\ifdim\ht0=\ht2
	\{\kern-.625\wd2 \{#1\}\kern-.625\wd2 \}%
	\else
	\left\{\kern-.7\wd0\left\{#1\right\}\kern-.7\wd0\right\}%
	\fi
}

\NewDocumentCommand{\dbktx}{om}{%
	\sbox0{\mathsurround=0pt$#1\{$}%
	\sbox2{\{}%
	\ifdim\ht0=\ht2
	\{\kern-.625\wd2 \{#2\}\kern-.625\wd2 \}%
	\else
	\mathopen{#1\{\kern-.7\wd0 #1\{}
	#2
	\mathclose{#1\}\kern-.7\wd0 #1\}}
	\fi
}
\title{Quantum Systems as Lie Algebroids}
\author{Ronald J. Ezuck\footnote{email: rjezuck@gmail.com}}

\begin{document}
	
	\maketitle
	
	\begin{abstract}
		Lie algebroids provide a natural medium to discuss classical systems, however, quantum systems have not been considered. It is the aim of this paper is to attempt to rectify this situation. Lie algebroids are reviewed and their use in classical systems is described. The geometric structure of the Schr\"{o}dinger and Heisenberg representations of quantum systems is examined and their relationship to Lie algebroids is explored. Geometrically, a quantum system is seen to be a collection of bounded, linear, self-adjoint operators on a Hilbert, or more precisely, a K\"{a}hler manifold. The geometry of the Heisenberg representation is given by the Poisson structure of the co-adjoint orbits on the dual of the Lie algebra. Finally, it is shown that the Schr\"{o}dinger and Heisenberg representations are equivalent.
		\\
		
	\end{abstract}

	\newpage
	
\section{Introduction}
Lie Algebroids have been used to study classical systems for some time \cite{online:grabowskadiracalgebroids,online:geommechanonalgebroids,online:hamiltdynoonliealgpreservvol,article:martinex1999,article:geoformogmechonliealgebroids,article:liealggeomandhamformal-popescu}. However, Lie algebroid formulations of quantum mechanics are notably absent. There have been approaches using deformation quantization with Lie algebroids \cite{online:quantpoissonalgassocliealgbroids-landsman,woodhouse_1994} but, otherwise, there has been little work done on the topic. This is unfortunate because much of the structure of quantum mechanics is revealed through Lie algebroids. The aim of this paper is to take some small steps to correct this situation and cast quantum mechanics as a Lie algebroid.
\\
\\
Loosely speaking, a Lie algebroid is a vector bundle over a manifold whose sections form a Lie algebra and possesses a bundle morphism, the anchor, that maps the sections onto the tangent bundle over the manifold \cite{book:mackenzieliegrpoidsliealgbroidsindiffgeom,book:mackenziegenthryliegrpoidsliealgbroids,article:pradines1,article:pradines2} and that satisfies the Leibniz rule.
\\
\\
In classical mechanics the approach is to consider the tangent bundle, in the case of Lagrange dynamics, or the cotangent bundle in the case of Hamiltonian mechanics \cite{book:abrahammarsdenfoundmech,online:hamiltliealgebroids-blomann-weinstein,online:lagrangesubmnfdsdynamonliealgebroids}. We will consider only Hamiltonian mechanics in this paper. The cotangent bundle forms a phase space consisting of the configuration and momentum space. This space also possesses a two-form that is symplectic. A state is then a point in this phase space. An observable is given by a function, the Hamiltonian, and is used to generate the sections of the bundle. These sections form a Lie algebra. The evolution is given by the anchor and is defined by Hamilton's equations. This, in turn, gives a vector space in the tangent bundle and generates a one-parameter group of diffeomorphisms, the flow. A measurement is obtained from the resulting integral distribution at a point.
\\
\\
In order to describe quantum mechanics within the framework of Lie algebroids, quantum mechanics must be given a geometric structure. Quantum mechanics is not normally thought of as a geometrical theory. Historically, the idea that classical and quantum systems could share the same complex coordinates was discussed by Strocchi \cite{article:complxcoordandqm-strocchi}. Kibble is considered as the first attempt to cast quantum mechanics as having a geometrical structure by introducing the notion of a quantum phase space \cite{article:geomofqm-kibble}. A more physical interpretation was given by Heslot \cite{article:qmasaclassthry-heslot}. Subsequent work begins with a quantum state as consisting of a vector in Hilbert space. Since two quantum states are equivalent up to a constant phase, $\phi \sim c \psi$, a physical quantum state is given by a line in a projective space. It is further observed that the Hilbert space can be recast as a K\"{a}hler manifold with a Hermitian inner product consisting of the sum of a symmetric inner product and an imaginary symplectic inner product \cite{online:geomformofqm,article:geomformofqm-clemente-gallardo,online:geomformofqm-heydari}.
\\
\\
There are a two standard ways of representing quantum mechanics that we will consider, the Schr\"{o}dinger and the Heisenberg pictures\footnote{There are more representations of quantum mechanics, for example the Dirac or interaction picture. This is a combination of the Schr\"{o}dinger and Heisenberg representations, and, for purposes of space, will not be discussed here.}. In the Schr\"{o}dinger representation the Hilbert space is replaced with a similar structure, a projective K\"{a}hler bundle. Observables are introduced as a constraint space of sections and form the kinematic foundation called the quantum phase space. An observable is given by an operator that acts on the quantum phase space and is given by the Hamiltonian. Taking the expectation value gives the Hamiltonian function. The evolution is determined by the anchor consisting of the Schr\"{o}dinger equation applied to the Hamiltonian function of the observable. The sections, in turn, form a Lie algebra and the bundle is principal. This generates an vector subspace of the tangent bundle with evolution operator given by the unitary transformations, $U_t = e^{iHt}$. Solutions are given by eigenvalues and eigenvectors of the observable operator.
\\
\\
In the Heisenberg representation, the emphasis is on the measurement value. The state remains fixed and the observables evolve. In this case, the observables form a $C^*-Algebra$ and act on the initial state. The geometrical structure is given by the orbits of the dual space of Hermitian operators  \cite{online:geomofqsys_denandentangl-grabowski-etal} with the Poisson product. This results in a principal fiber bundle over the unitary group. The Hamiltonian sections are given by the momentum (or moment) map. The evolution is given by the Heisenberg equation and generates a Heisenberg vector space. A measurement is obtained by calculating the eigenvalues and eigenvectors from this vector space.
\\
\\
A little more detail concerning observables is in order. In classical systems, a state is a point in the phase space and an observable is a function, $f$, that selects a particular state, $f(q,p) \in T(Q)$. In a sense, a classical observable has the measurement apparatus built in and the result is obtained with certainty. This can be attributed to the ability of classical systems to be measured by non-invasive measurement apparatus \cite{inprep:measspaceinliealgeb}. In quantum systems, for example in the Schr\"{o}dinger representation, an observable is a self-adjoint, linear operator representing a measurement apparatus that takes a state to a state with a probability. The measurement apparatus in the quantum case cannot be minimized. The inclusion of the measurement process is one of the characteristic difference between classical and quantum mechanics.
\\
\\
The organization of the paper is as follows. We will consider pure states in a finite-dimensional Hilbert space. This is not a severe restriction as mixed states are combinations of pure states. In the second section, the general properties of Lie algebroids are reviewed. Classical Hamiltonian mechanics is described on the prolongation of a Lie algebroid.
\\
\\
The third section sets the stage to place a quantum system within the Lie algebroid structure. We lay the groundwork by reviewing the geometric construction of the Schr\"{o}dinger and Heisenberg representations of quantum mechanics. Quantum mechanics is defined on a projective Hilbert space which is expressed as a K\"{a}hler manifold. The state space is then represented as a principal fiber bundle and provides the kinematic foundation. This process generates the quantum phase space. In the case of the Schr\"{o}dinger representation, the anchor is given by Schr\"{o}dinger equation and takes sections of the projective space to a Hilbert space and describes the evolution. The evolution operator is given by unitary transformations, $U_t = e^{iHt}$. Solutions are obtained by calculating the eigenvalues of the observable operator.
\\
\\
In the Heisenberg representation, the fibers are bounded, self-adjoint operators that form a $C^*$-algebra space and gives rise to a Lie algebra generated by the Unitary group. The geometry is obtained from the Poisson structure of the dual Lie algebra. The momentum (moment) map is the Hamiltonian and introduces the Lie algebra sections. The evolution of the system is given by the Heisenberg equation and generates a vector field contained in the associated tangent bundle with the measurement values obtained by calculating the eigenvalues and eigenvectors. The equivalence of the Schr\'{o}diger and Heisenberg interpretations is demonstrated. 
\\
\\
The final section discusses the results and explores next steps.
\section{Lie Algebroids}
Lie algebroids nicely delineates the structure of classical Hamiltonian processes. The cotangent bundle defines the kinematic structure by establishing the state space and defining the constraint space of observables. The sections are given by the Hamiltonian and form a natural Lie algebraic structure. The anchor specifies the evolution the system and provides for the measurement of a particular state.
\\
\\
We begin with some definitions. A Lie algebroid can be thought of as an extension of a tangent bundle. Formally, we have the definition \cite{book:cannasweinsteingeomodfornoncommalg,article:corteslangrangemechonliealgebroids,article:liberman_1996,book:mackenzieliegrpoidsliealgbroidsindiffgeom,book:mackenziegenthryliegrpoidsliealgbroids,article:langrangmechonliealgbroids,confproc:martinezliealgebroidinclassmech,article:weinsteinlagrngdynandliegrpds}:
\begin{defn}
	Let $\pi:E \rightarrow M$ be a vector bundle with $\Gamma[E]$ the global sections that form a Lie algebra, $\mathcal{A}$. $E$ is the total space with vector space fibers, $M$ is the base space, and $\pi : E \rightarrow M$, the standard projection. An \textbf{almost Lie Algebroid} \textnormal{\cite{online:pelletiergeomstructonprolong}}, $(E, \llbracket \cdot , \cdot \rrbracket , \rho)$, consists of a vector bundle and a map $\rho : \Gamma[E] \rightarrow \mathfrak{X} \subset T M$ called the \textbf{anchor map}, from the global sections, $\Gamma[E] = \mathcal{A}$ to the tangent bundle over $M$, $TM$, such that, for $v,w \in \Gamma[E]$:
	\begin{equation*}
		\rho \llbracket v,w \rrbracket = \llbracket \rho v, \rho w \rrbracket 
	\end{equation*}
	\begin{equation*}
	and \; 
	for \; f \in C^{\infty}, \; \llbracket v,fw \rrbracket = \rho(v) fw + f\llbracket v, w \rrbracket \quad (Leibniz Rule)
\end{equation*}
\end{defn}
\begin{defn}
	Let $\pi : E \rightarrow M$ be an almost Lie Algebroid. If it also satisfies the Jacobi identity
	\begin{equation*}
	\llbracket u,\llbracket v,w \rrbracket \rrbracket + \llbracket v,\llbracket w,u \rrbracket \rrbracket + \llbracket w, \llbracket u,v \rrbracket \rrbracket = 0
	\end{equation*}
	then it is a \textbf{Lie Algebroid}.
\end{defn}
\begin{rmk}
	We note that the anchor is a bundle homomorphism, $\rho : \Gamma(E) \rightarrow \mathfrak{X}$ between Lie Algebras, that is, $\rho : (\Gamma(E) , \llbracket \cdot , \cdot \rrbracket) \rightarrow (\mathfrak{X}(M), [ \cdot, \cdot ])$ \cite{online:lagrangesubmnfdsdynamonliealgebroids}.
\end{rmk}
\begin{rmk}
	If $\rho$ is fiber-wise surjective, $\rho$ is said to be \textbf{transitive}. If $\rho = 0$, $\rho$ is  \textbf{totally intransitive}. If the rank of $\rho$ is locally constant, $\rho$ is \textbf{regular}.
\end{rmk}
\begin{rmk}
	Let $\{e_{\alpha}\}$ be a local basis for the sections ($y = y^{\alpha} e_{\alpha}$) and $\{x^i\}$ a local coordinate system on $M$. Then $(x^i, y_{\alpha})$ forms a coordinate system on $E$. The anchor and the bracket are given by
	\begin{equation*}
	\rho(e_{\alpha}) = \rho^i_{\alpha} \frac{\partial}{\partial x^i} \qquad and \qquad \llbracket e_{\alpha}, e_{\beta} \rrbracket = C^{\gamma}_{\alpha \beta} e_{\gamma}
	\end{equation*}
	Applying the Leibniz and Jacobi identity conditions we have
	\begin{align*}\label{lastruct}
	\rho^j_{\alpha} \frac{\partial \rho^i _{\beta}}{\partial x^j} - &\rho^j_{\beta} \frac{\partial \rho^i _{\alpha}}{\partial x^j} = \rho^i_{\gamma} C^{\gamma}_{\alpha \beta}\\
	\rho^i_{\alpha} \frac{\partial C^{\nu}_{\beta \gamma}}{\partial x^i} + \rho^i_{\beta} \frac{\partial C^{\nu}_{\gamma \alpha}}{\partial x^i} + \rho^i_{\gamma} \frac{\partial C^{\nu}_{\alpha \beta}}{\partial x^i}& + C^{\mu}_{\beta \gamma}C^{\nu}_{\alpha \mu} + C^{\mu}_{\gamma \alpha}C^{\nu}_{\beta \mu} + C^{\mu}_{\alpha \beta}C^{\nu}_{\gamma \mu} = 0
	\end{align*}
	These are referred to as the structure equations.
\end{rmk}
\begin{expl}\textbf{Lie Algebra}\\
	A Lie Algebra over a point, $M = \{p\}$, is a Lie Algebroid.
\end{expl}
\begin{expl}\textbf{Tangent bundle}
	\\
	Let $TM = (TM,\llbracket \cdot, \cdot \rrbracket, \rho)$ and $\rho(x) = id$. Then it is trivially a Lie Algebroid. 
\end{expl}
\begin{expl}\textbf{Atiyah Lie Algebroid}
	\\
	Let $P$ be a principal bundle over a manifold $M$ with structure group, $G$. Then $TP / G$ is a Lie algebroid with a surjective anchor, $\rho$. The sections form a Lie algebra of $G$-invariant vector fields, $\mathfrak{g}$.
\end{expl}
\noindent
We have some standard constructions (\cite{online:lagrangesubmnfdsdynamonliealgebroids,online:pelletiergeomstructonprolong} and previous references).
\begin{defn}\label{eq:algdiff}
	Let (E,$\llbracket \cdot , \cdot \rrbracket$, $\rho$) be a Lie algebroid, $E_p$ a fiber over $p$. The \textbf{differential} of $E$, $d^E : \Gamma(\bigwedge^r E^*) \rightarrow \Gamma(\bigwedge^{r+1} E^*)$ is given by
	\begin{align}
	d^E \mu (X_0, \cdots , X_r) &= \sum_{i = 0}^{r} (-1)^i \rho(X_i)(\mu (X_0, \cdots , \hat{X_i}, \cdots , X_r)) \nonumber\\
	&+ \sum_{i < j} \mu (\llbracket X_i, X_j \rrbracket , X_0 , \cdots , \hat{X_i}, \cdots , \hat{X_j} , \cdots, X_r)
	\end{align}
	where $\mu \in \Gamma(\bigwedge^r E^*)$ and $X_0, \cdots  X_r \in \Gamma(E)$. The hat, $\hat{X_i}$, denotes the omission of $X_i$. We note that $(d^E)^2 = 0$.
\end{defn}
\noindent
If we have local coordinates we get
\begin{equation}
dx^i = \rho^i_{\alpha} e^{\alpha} \quad and \quad de^{\alpha} = -\frac{1}{2} C^{\alpha}_{\beta \gamma} e^{\beta} \wedge e^{\gamma}
\end{equation}
where $\{e^{\alpha}\}$ the dual of $\{e_{\alpha}\}$, the section basis.
\begin{defn}
	The \textbf{Lie derivative} of a function, $f$, on a Lie algebroid is given by $\mathcal{L}_X f = Xf$ \textnormal{\cite{book:frankelgeoofphysintro}}. The Lie derivative of a vector, $Y$, is $\mathcal{L}_X Y = [X,Y]$.
	\\
	\\
	The Lie derivative, $\mathcal{L} : \Gamma(\bigwedge^r E^*) \rightarrow \Gamma(\bigwedge^r E^*)$, of a differential form is given by
	\begin{align*}
	\mathcal{L}_X^E &= i_X \circ d^E + d^E \circ i_X\\
	&= X \rfloor d^E + d^E \rfloor X
	\end{align*}
\end{defn}
\noindent
We note that $\Gamma(E)$ is involutive, that is, for $X, Y \in \Gamma(E),\; \left[X,Y\right] \in \Gamma(E)$. Thus from the Froebenius Theorem we get a integrable submanifold, or \emph{integrable distribution}. The image is also a smooth involutive distribution \cite{online:mechcontrlsysonliealg}. This also provides a natural foliation on M \cite{online:lagrangesubmnfdsdynamonliealgebroids,article:sussmanvectfldsandintdist}. The dual bundle fibers admit a Poisson structure \cite{online:calconliealgliegrpandpoismnflds}, $\{\cdot , \cdot\} : C^{\infty}(E^*) \times C^{\infty}(E^*) \rightarrow C^{\infty}(E^*)$ and having the properties
\begin{align*}
\{f,g\} &= - \{g,f\}\\
\{fg,h\} &= f\{g,h\} + g\{f,h\}\\
\{f,\{g,h\} + &\{g,\{h,f\} + \{h,\{f,g\} = 0
\end{align*}
where $f,g,h \in C^{\infty}(E^*)$.
\\
\\
If $(q_i, p^j)$ are local coordinates where $\{e^j\}$ is a local basis for $E^*$ we have
\begin{equation*}
\{f,g\} = \rho^i_j\left(\frac{\partial f}{\partial q_i} \frac{\partial g}{\partial p^j} - \frac{\partial g}{\partial q_i}\frac{\partial f}{\partial p^j} \right) - C_{jk}^l p^l \frac{\partial f}{\partial p^j}\frac{\partial g}{\partial p^k}
\end{equation*}
\subsection{Prolongations}
One way Lie algebroids can be extended to satisfy a wider variety of applications is by defining a prolongation \cite{online:lagrangesubmnfdsdynamonliealgebroids,article:leokDiracStructAndHamJThryonLieAlg,book:mackenziegenthryliegrpoidsliealgbroids,article:martinex1999,article:geoformogmechonliealgebroids,confproc:martinezliealgebroidinclassmech,article:dualstruct-prolong-lalbroid-popescu,article:prolongofliegroupoidsalgebroids}. This will allow us to define second order differential equations (SODE).
\\
\\
Suppose we have a Lie algebroid with vector bundle, $\tau: E \rightarrow M$ and an anchor map $\rho: \Gamma(E) \rightarrow TM$. Suppose, further, that we have a map $f: M' \rightarrow M$ such that $TM' \rightarrow M'$. Now, we can define $E \times TM'$ such that, for $(b,v), \; b \in \Gamma(E), \; v \in TM', \; \rho(b) = T_f(v)$, we will call this a prolongation \cite{article:higginsmackenziealgconstinlieaglebroids}. A prolongation can be thought of as a tangent bundle to the total space $M \times E$. Formally, we have
\begin{defn}\label{def:prolongspace}
	Let  $\pi: E \rightarrow M$ be a Lie Algebroid and let $\tau : P \rightarrow M$ be a fibration. For $p \in P$ define a vector space by
	\begin{equation*}
	\mathfrak{L}P_p = T_p^E P = \{(b,v) \in E_x \times T_p P \; \vert \; \rho(b) = T_p \tau(v)\}
	\end{equation*}
	$T_p : TP \rightarrow TM$ is the tangent map of $P$ at $p$. This is the fiber of the \textbf{prolongation} bundle. Here  $\mathfrak{L}P$ is a bundle over $P$.
\end{defn}
Notice that the tangent is with respect to the point p. In effect, we have
\begin{equation*}
	\mathfrak{L}P_p = T_p^E P = \{(p,b,v) \in E_p \times E_x \times T_p P \; \vert \; \rho(p) = \rho (b), \; \rho(b) = T_p \tau(v)\}
\end{equation*}
where $v \in T_pE$. This forms a vector bundle, $\mathcal{T}^E P = \bigcup_{p \in P} \mathfrak{L}P_p$. The projection is
given by $\pi_1(p,b,v) = p$.
\begin{defn}\label{def:prolongproj}
	We can define $\pi_2 : \mathfrak{L}E \rightarrow E$ as the projection onto the second factor, i.e., $\pi_2(p,b,v) = b$.  The projection of the third factor, $\pi_3(p,b,v)$, is the anchor of the Lie algebroid, $\mathfrak{L}E \rightarrow \mathcal{T}E,\; \rho(p,b,v) = v$. Finally we have $\pi_{12}: (p,b,v) \mapsto (p,b)$.
\end{defn}
	A section of the Lie algebroid vector bundle, $\Sigma$ of $\mathcal{T}^{E}P$, is given by 
\begin{equation*}
	\Sigma(p) = (p, \sigma(\tau(p)), X(p))
\end{equation*}
where $\sigma$ is a section of $E$.
For two sections, $\Sigma_1, \Sigma_2 \in \mathcal{T}^{E}P$, the Lie product is
\begin{equation*}
	[\Sigma_1, \Sigma_2] = (p, [\sigma_1, \sigma_2](p),[X_1,X_2](p)) ,\quad p \in E
\end{equation*}
A section of $\mathcal{T}E$, $\Sigma \in \mathcal{T}E$, is said to be \textbf{projectable} if it is projected onto a section $\sigma \in E$. That is, $\pi_2 \circ \Sigma = \sigma \circ \pi$.
We have then:
\begin{prop}
	The prolongation forms a Lie Algebroid.
\end{prop}
\begin{proof}
	From of above definition (Def. \ref{def:prolongspace}) $\mathcal{T}^E P = \bigcup_{p \in P} \mathfrak{L}P_p$ forms a bundle with $\Sigma$ forming a Lie algebra.\\
	\\
	The anchor is given by $\rho : \mathcal{T}^{E}P \rightarrow TE$ where $\rho(p,b,v) = \pi_3(p,b,v) = v$.
\end{proof}
\begin{expl}
	Suppose $E$ is the tangent bundle, $TE$, then $\mathcal{T}^EE$ is isomorphic to $T(TE)$.
\end{expl}
\noindent
\begin{rmk}
	A prolongation is a form of double bundle \cite{book:mackenziegenthryliegrpoidsliealgbroids}.
\end{rmk}
\begin{rmk}
	In Lagrange dynamics $\mathfrak{L}$ takes the form $\tau : T(TM) \rightarrow TM$ \cite{article:langrangmechonliealgbroids}.
\end{rmk}
\noindent
Let $(x^i, y^j)$ be local coordinates on $P$ with a local basis $\{e_j\}$ of sections in $E$. We can define a local basis $\{ \tilde{X}_{\alpha}, \tilde{Y}_{A}\}$ for sections of $\mathcal{T}^E P$ by
\begin{equation}\label{eq:prolongcoord}
\tilde{X}_{\alpha}(p) = \left(p, e_{\alpha}(\pi(p)), \rho^i_{\alpha} \frac{\partial}{\partial x^i}\bigg\vert_p\right) \quad and \quad \tilde{Y}_{A} (p) = \left( p, 0, \frac{\partial}{\partial y^{A}}\bigg\vert_p \right)
\end{equation}
We have the local coordinates ($(x^i, y^j)$ for $p$ and $(x^i, u^j)$ for b. If $z = (p,b,v) \in \mathcal{T}^E P$ where $b = z^{\alpha} e_{\alpha}$ we have for $v$
\begin{equation}
v = \rho^{i}_{\alpha} z^{\alpha} \frac{\partial}{\partial x^i} + v^{A} \frac{\partial}{\partial y^{A}}
\end{equation}
and $z$ has the form
\begin{equation}
z = z^{\alpha} x^{'}_{\alpha} (p) + v^A v^{'}_A (p)
\end{equation}
From this, a vertical element is a linear combination of $\{ v^{'}_A\}$.
\\
\\
Now, if $\eta$ is a section of the prolongation $T^EP$ with local coordinates $(x^i, y^j, u^j(x,y), z^j(x,y))$ it can be expressed as
\begin{equation*}
\eta = u^j \tilde{X}^j + z^j \tilde{Y}^j
\end{equation*}
The anchor map takes this to the vector space given by
\begin{equation*}
\rho(z) = \rho^i_j u^j(x,y) \frac{\partial}{\partial x^i}\bigg\vert_{(x,y)} + z^j(x,y) \frac{\partial}{\partial y^j}\bigg\vert_{(x,y)}
\end{equation*}
We have a number of structures on a prolongation.
\begin{defn}
	Let $z \in T_p^E P$,
	 $z$ is said to be \textbf{vertical} if it projects to zero. That is, $T\pi(z) = (p,0,v)$, where $v$ is the vertical vector at $T_p$. Let $f \in C^{\infty}$. The \textbf{vertical lift} is given by $f^v = f(\pi(w))$. The \textbf{complete lift} is given by $f^c (w) = \rho(w)(f)$, where $w \in E$. Similarly, for $\sigma$ a section in $E$. The vertical lift of $\sigma$ is the vector field $X$ on $E$ given by $X^v(w) = X(\pi(w))^v_w , w \in E$ \normalfont{\cite{article:geoformogmechonliealgebroids,online:vertcompliftslielagebroids-peyghan}}.
\end{defn}
\noindent
\begin{defn}
	A tangent vector, $v$ at $p$ is said to be \textbf{admissible} if $T_p \pi(v) = \tau(p)$. A curve is admissible if its tangent vectors are admissible.
\end{defn}
For a general Lie algebroid, we can replace the notion of a prolongation with that of an admissible curve \cite{article:langrangmechonliealgbroids}.

\subsection{Hamiltonian Formalism on Lie Algebroids}\label{sect:liealgham}
Recall that a Hamiltonian system is a cotangent bundle, $T^*(Q)$, over some configuration space $Q$. The cotangent bundle represents the possible states that a system can be in. This space forms a symplectic manifold through its Poisson structure that generates a non-degenerate two-form, called a symplectic form \cite{mcduff_salamon_1995}, $\omega$, such that $d \omega = 0$, that is, it is closed. In the case where the adapted coordinates of the cotangent bundle are $(q,p)$ with $\omega = \omega(q,p) = \omega_{ij}dq^i \wedge dp^j$ we have the \textbf{canonical phase space}, $\mathcal{M}$. The state of a classical system is a point in this canonical phase space.
\\
\\
Consider the Hamiltonian given by $H = \frac{p^2}{2m} + V(q)$. The equations of motion are given by
\begin{equation}
	\dot{q} = \frac {\partial H}{\partial p} \quad \dot{p} = - \frac {\partial H}{\partial q}
\end{equation}
The Poisson bracket is
\begin{equation}
	\{f,g\} = \frac{\partial f}{\partial q} \frac{\partial g}{\partial p} - \frac{\partial g}{\partial q} \frac{\partial f}{\partial p}.
\end{equation}
We have $\dot{x} = \{H,x\}$ for a Hamiltonian $H$ and Hamilton's equations become
\begin{equation}\label{eq:poissonham}
	\dot{q} = \{q,H\} \quad and \quad \dot{p} = \{p,H\}.\\
\end{equation}
This gives the evolution of the system and, as we shall see, is the anchor.
\begin{defn} Let $TQ$ be the canonical tangent (vector) space and $T^* Q$ the canonical cotangent space \textnormal{\cite{online:sardanashvilyintHamsys}}. Let $\omega(u,v)$ be a two-form. The kernel of the map, $ker\; \omega$, contains the zero-section, $\hat{\textbf{0}}$, of the tangent bundle $TQ \rightarrow Q$. If $ker\; \omega = \hat{\textbf{0}}$ then $\omega$ is non-degenerate. If, in addition, $\omega$ is closed, $d\omega = 0$ and $\omega$ is said to be \textit{symplectic}. $(M,\omega)$ with $\omega$ closed and non-degenerate is a \textbf{symplectic manifold}.
\end{defn}
An observable is given by a Hamiltonian that produces a moment map, $H: Q \rightarrow \mathfrak{g}^*$ from the configuration space to the dual of the Lie algebra of sections. A moment map generates a vector field defined by $X_H \rfloor \omega = dH$. This is the \textit{Hamiltonian vector field}. The Hamiltonian vector field, in turn, generates a one-parameter group of diffeomorphisms, $g_t$, called the \textit{flow} that describes the evolution of the system. Thus, a Hamiltonian selects a subspace of the total space.
\begin{defn}
	\textnormal{Let $\sigma$ be a section on $T^*Q$. If $\sigma_f \rfloor \omega$ is exact, then $\sigma_f$ is called \textbf{Hamiltonian}. Thus, for $f \in C(M)^{\infty}$, $\sigma_{f} \rfloor \omega = -df$ is a \textbf{Hamiltonian section} of $T^*Q$ with respect to $f$. Furthermore, since $T^*Q$ is a vector bundle, $\sigma_H$ is a global section. $H : T^*Q \rightarrow \mathbb{R}$ is called the \textbf{Hamiltonian function}. Applying the anchor we get the \textbf{Hamiltonian vector field} $X_H = \rho(\sigma_H)$ in the associated tangent bundle $TQ$.}
\end{defn}
On a canonical phase space this has a natural extension as a Lie Algebroid with the anchor map, $\rho : \Gamma(X) \rightarrow T(Q)$. As we will see, Hamilton's equations form the anchor map.
\\
\\
\noindent
\textbf{Example:} Let $H : T^*Q \rightarrow \mathbb{R}$ be a Hamiltonian defined on the canonical basis given by $(q_i,p^i)$, and define the \textit{Liouville} form by $\theta = p^i dq_i$. The canonical form is $d \theta = dp^i \wedge dq_j$. Then the Hamilton vector field is given by
\begin{align*}
X_H &= \{H, \cdot\}\\
 &= \partial^i H \partial_i - \partial_i H \partial^i \nonumber\\
\end{align*}
In phase space, $(p^i, q_i)$
\begin{equation*}
X_H = \sum_{i=1}^{n} \frac{\partial H}{\partial p_i} \frac{\partial}{\partial q_i} - \frac{\partial H}{\partial q_i} \frac{\partial}{\partial p_i} \label{eq:Hamvect}.
\end{equation*}
If, for $u,v \in X_H$, $\mathcal{L}_u v = 0$, then $X_H$ is the infinitesimal generator of a one-parameter group of diffeomorphisms and $X_H$ is said to be an \textit{integral manifold}. It is also a \textit{distribution}, and so, by Frobenius' Theorem forms a foliation, $\mathcal{F}$.
\\
\\
\noindent
We wish to describe this via a prolongation on a Lie Algebroid. To this end we note that a Hamiltonian system is defined on the dual to a prolongation, $T(T^*M)$ \cite{book:abrahammarsdenfoundmech,book:cartangeomandtheirsymm-crampin,article:corteslangrangemechonliealgebroids,online:grabowskadiracalgebroids,online:geommechanonalgebroids,article:liberman_1996,article:liealggeomandhamformal-popescu,article:weinsteinlagrngdynandliegrpds}.
\begin{defn}\label{def:symplcliealg}
	Let $(E^*, \llbracket \cdot, \cdot \rrbracket, \rho)$ be the dual to a Lie algebroid over a manifold, $M$, that has sections given by a non-degenerate, closed two-form, $\omega$. Such a section is called a \textbf{symplectic section}. A Lie algebroid possessing a symplectic section is called a \textbf{symplectic Lie algebroid}. 
\end{defn}
\noindent
Suppose we have a dual Lie Algebroid $(E^*, \llbracket \cdot , \cdot \rrbracket, \rho)$ and let $\tau$ be the bundle projection from $\tau : E^* \rightarrow M$. Suppose, further, that $\Gamma[E^*]$ is a $C^{\infty}(M)$-module of symplectic sections and let $\rho : \Gamma(E^*) \rightarrow \mathfrak{X}(M) \subset TM$ be the anchor map \cite{online:lagrangesubmnfdsdynamonliealgebroids}. Such a Lie Algebroid is said to be symplectic if the sections have a non-degenerate closed two-form, $\omega(x,y)$. Let $\mathfrak{L}^{\tau}_{E^*} = (\mathcal{T}E^*, \llbracket \cdot, \cdot \rrbracket, \rho)$ be a prolongation on $E^*$. We note that on a Lie algebroid, for $\omega \in \Gamma(E^*)$, $\omega$ is a non-degenerate two-form and that $d^{\mathfrak{L}^{\tau^*}_{E^*}} \omega = 0$ thus we have the following:
\begin{thrm}\label{thm:symprol}
	Let $(E^*, \llbracket \cdot , \cdot \rrbracket, \rho)$ be a dual Lie algebroid and $\mathfrak{L}^{\tau}_{E^*}$ be the associated prolongation. Then, if the Lie algebroid is symplectic, the prolongation is symplectic as well.
\end{thrm}
\begin{proof}
	This follows immediately from the definitions (above) and the fact that the prolongation is defined on a symplectic section of the Lie algebroid.
\end{proof}
Let $(E,\pi,\rho)$ be a Lie algebroid defined over coordinates $(q,p)$. From our discussion above, it possesses a prolongation. The dual to the Lie algebroid, $(E^*,\tau,\rho)$ also admits a prolongation, $\tau_1 : \mathcal{T}^E E^* \rightarrow E^*$ that forms a Lie algebroid \cite{article:corteslangrangemechonliealgebroids,article:geoformogmechonliealgebroids}. Analogously to Definition \ref{def:prolongspace} the total space of the dual prolongation at a point $a$ is given by
\begin{equation*}
\mathcal{T}^EE = \{ (a,b,w) \in E^{*} \times E \times T^{*}E | \tau(a) = \pi(b), w \in T^{*}_aE \; and \; \rho(b) = T_a \tau(w)\}
\end{equation*}
the projection is given by $\tau_1(a, b, w) = a$. The sections have a basis given by (see Eq.\ref{eq:prolongcoord})
\begin{equation}
\tilde{X}_j(a) = \left( a, e_j, \rho^i_j \frac{\partial}{\partial q^i} \bigg\vert_a \right) \quad and \quad \tilde{V}_j(a) = \left( a, 0, \frac{\partial}{\partial p_j}\bigg\vert_a\right)
\end{equation}
where $(q^i,  p^j)$ are the coordinates form a section with basis $\{e_j\}$ \cite{article:geoformgeomechonliealhebroids}. Notice that $\tilde{V}$ is vertical.
\\
\\
From the definition (\ref{def:prolongproj}) the anchor map is given by $\rho(a,b,\omega) = w$. The Lie algebroid brackets are given by
\begin{equation*}
\llbracket \tilde{X}_i , \tilde{X}_j \rrbracket = C^k_{ij} \tilde{X}_k \quad \llbracket \tilde{X}_i, \tilde{V}^j \rrbracket = 0 \quad \llbracket \tilde{Y}_i, \tilde{V}^{j} \rrbracket = 0 
\end{equation*}
The vertical lift and the complete lift provide a natural basis for a section on the prolongation. Let $\sigma = \sigma^i e_i$ and $\theta_i e^i$ be the sections of $E$ and $E^*$ respectively. The vertical lift is given by $\theta^V = \theta_i \tilde{V}^{j}$, and the complete lift is
\begin{equation*}
\sigma^C = \sigma^i \tilde{X}_{i} - q_k \left( \rho^k_i \frac{\partial \sigma^k}{\partial q_k} + C^k_{ij} \sigma^j \right) \tilde{V}^{i}
\end{equation*}
Let $\eta$ be another complete lift, the brackets have the form
\begin{equation*}
\llbracket \sigma^C , \eta^C \rrbracket = \llbracket \sigma , \eta \rrbracket^C \quad \llbracket \sigma^C, \theta^V \rrbracket = (d_{\sigma} \theta)^V \quad \llbracket \tilde{V}_i, \theta^{V} \rrbracket = 0 
\end{equation*}
We can define a symplectic form on $\mathcal{T}^EE^{*} \rightarrow E^{*}$ by 
\begin{equation*}
\langle \theta_{0a}, (a,b,\omega)\rangle = \langle a,b \rangle
\end{equation*}
and $\omega_0 = -d\theta_0$. These have coordinates
\begin{equation*}
\theta_0 = q_i \tilde{X}^{i} \quad and \quad \omega_0 = \tilde{X}^{i}  \wedge \tilde{V}_i + \frac{1}{2} q_k C^k_{ij} \tilde{X}^{i} \wedge \tilde{X}^{j}
\end{equation*}
We now relate the anchor with the generation of a Hamiltonian vector field. Recall the definition of a Hamiltonian vector field.
\begin{defn}\label{def:hamsect}
	Suppose we have a function, $H: T^*Q \rightarrow \mathbb{R}$. A \textbf{Hamiltonian section}, $\sigma_{H}$, is then given by
	\begin{equation}
	\sigma_H \rfloor \omega = \iota_{\sigma_{H}} \omega = dH
	\end{equation}
	where $\omega$ is the symplectic two-form and $d$ is the differential operator \textnormal{\cite{article:corteslangrangemechonliealgebroids}}.
\end{defn}
\begin{rmk}
	For a Hamiltonian section, the \textbf{Hamiltonian vector field} is given by the anchor map
	\begin{equation}
	X_{H} = \rho(\sigma_{H})
	\end{equation}
\end{rmk}
\noindent
Applying definition (\ref{eq:algdiff}) to local coordinates we have
\begin{align*}
d^{E}f =& \frac{\partial f}{\partial q^i} \rho^{i}_{\alpha} e^{\alpha}\\
d^{E}\theta =& \left(\frac{\partial \theta_{\gamma}}{\partial q^{i}} \rho^{i}_{\beta} - \frac{1}{2} \theta_{\alpha} C^{\alpha}_{\beta \gamma}\right) e^{\beta} \wedge e^{\gamma}\\
\textnormal{thus} &\\
d^{E} q^{i} =& \rho^{i}_{\alpha} e^{\alpha},  \quad d^{E} e^{\alpha} = - \frac{1}{2} C^{\alpha}_{\beta \gamma} e^{\beta} \wedge e^{\gamma}\\
\end{align*}
For a function $H$ on $E^{*}$ we have a Hamiltonian section $\sigma_H$ of $\mathcal{T}^E E^{*}$ given by
\begin{equation*}
\sigma_H \rfloor \omega_0 = \iota_{\sigma_{H}}\omega_0 = \omega(\sigma_{H}, \cdot) = dH
\end{equation*}
the section is given by
\begin{equation*}
\sigma_H = \frac{\partial H}{\partial p_i} \tilde{X}_i - \left( \rho^k_i \frac{\partial H}{\partial q^k} + p_k C^k_{ij} \frac{\partial H}{\partial p_j} \right) \tilde{V}^{i}
\end{equation*}
and the anchor gives a (Hamiltonian) vector field $X_H = \rho(\sigma_H) \in \mathfrak{X}(E^{*})$,
\begin{equation*}
\rho(\sigma_H) = \rho^i_k \frac{\partial H}{\partial p_k} \frac{\partial}{\partial q_i} - \left( C^k_{ij}p_k \frac{\partial H}{\partial p_j} + \rho^k_i \frac{\partial H}{\partial q_k} \right) \frac{\partial}{\partial p_i}
\end{equation*}
The equations for the integral curves are then given by
\begin{equation}\label{eq:LieAlgHam}
\frac{dq^i}{dt} = \rho^i_k \frac{\partial H}{\partial p_k} \qquad \frac{dp_i}{dt} = - \rho^k_i \frac{\partial H}{\partial q^k} - p_k C^k_{ij} \frac{\partial H}{\partial p_j}
\end{equation}
The Poisson bracket has the form $\{F,G\} = -\omega (\sigma_F, \sigma_G) $.
\\
\\
When we restrict the underlying vector bundle to the canonical phase space we get the standard results as demonstrated in the following proposition:
\begin{prop}\label{prop:hamprolonphsspc}
	Let (q,p) be the canonical coordinates of the phase space of mechanics, $(x, y) \rightarrow (q,p)$, $(E, \llbracket \cdot, \cdot\rrbracket, \rho) \rightarrow (T^*Q, [\cdot,\cdot], id)$ and $H : T^*Q \rightarrow \mathbb{R}$ then we have
	\begin{equation}\label{eq:liealgbrdham2}
	\dot{q_i} = \frac{\partial H}{\partial p_i} \quad and \quad \dot{p_i} = - \frac{\partial H}{\partial q_i}
	\end{equation}
\end{prop}
\begin{proof}
	Since $E = T^*Q$ and $ H: T^*Q \longrightarrow \mathbb{R}$, from equation \ref{eq:LieAlgHam}, where $\rho = id$ and $C^k_{ij} = 0$. We have then
	\begin{align*}
	\dot{q_i} = \frac{dq_i}{dt} &= \rho^k_i \frac{\partial H}{\partial p_k}\\
	&= \frac{\partial H}{\partial p_i}
	\end{align*}
	and
	\begin{align*}
	\dot{p_i} = \frac{dp_i}{dt} &= - \rho^k_i \frac{\partial H}{\partial q^k} - p_k C^k_{ij} \frac{\partial H}{\partial p_j}\\
	&= - \frac{\partial H}{\partial q_i}
	\end{align*}
\end{proof}
\noindent
Thus the anchor map generates the Hamiltonian vector field through the Hamiltonian equations. We see that the use of Hamiltonian Lie Algebroids allows the solution space to be highlighted. This will be said to define a measurement space.
\\
\\

\section{Quantum Systems on Lie Algebroids}
Since its inception, quantum mechanics was formulated as an algebraic system and little has changed by advances in the mathematics since then. The prevailing attitude has been the ``shut up and calculate'' approach\footnote{The expression ``Shut up and calculate'' is usually attributed to Feynman although there is some evidence to suggest that it is, in fact, due to Mermin \cite{article:feynmanquote-mermin}.}. In order to properly express quantum mechanics in the Lie algebroid environment we will describe it as a geometric system \cite{online:geomformofqm,incoll:somermksonhamsysandqm-chernoff-marsden,article:geomaspctsqmandqentangl-chruciski,article:geomformofqm-clemente-gallardo,book:geomqm-geroch,online:geomformofqm-heydari,article:geomofqm-kibble,thesis:geomofqm-schilling, article:geommethinqm-spera}.
\\
\\
Further, the quantum case is somewhat more involved as we have the question regarding its representation. The Scr\"{o}dinger representation places the emphasis on the state of the system. In the Heisenberg representation, the emphasis is on the physical magnitude of the measurement via the observables. As we will show, the two representations are equivalent.
\\
\\
The basic structure of a quantum system is a wave function. Since we cannot measure a wave function directly we measure a particular property of the system. We do this by applying an operator to a quantum state and get its value. By Born's Rule, this results in a probability. Consider the expectation value for an operator, $\hat{A}$
\begin{align}
\langle \hat{A} \rangle &= \langle \psi_t \vert \hat{A} \vert \psi_t \rangle \label{eq:opexpect} \\
	&=\langle \psi_0 \vert U^{\dagger}_t \hat{A} U_t \vert \psi_0 \rangle
\end{align}
We have two ways of looking at this, we have the view
\begin{equation}
\vert \psi_t \rangle = \hat{A} \vert \psi_0 \rangle
\end{equation}
in which case we get the Schr\"{o}dinger representation
\begin{equation}
\langle \hat{A} \rangle = ( \langle \psi_0 \vert U^{\dagger} ) \hat{A}  (U \vert \psi_0 \rangle)
\end{equation}
We can also fix the initial wave function and vary the operator to get the Heisenberg representation
\begin{equation}
\langle \hat{A} \rangle = \langle \psi_0 \vert (U^{\dagger} \hat{A} U) \vert \psi_0 \rangle
\end{equation}
We will discuss each of these representations below.
\subsection{The Schr\"{o}dinger Representation }
In this section we will consider pure states of a finite-dimensional Hilbert space in the Schr\"{o}dinger picture. The Schr\"{o}dinger representation is characterized by
\begin{enumerate}
	\item The state of a quantum system is a vector in Hilbert space.
	\item An observable is a linear, self-adjoint operator acting on a Hilbert vector.
	\item The state $| \psi \rangle$ changes with time through the introduction of a unitary operator $| \psi \rangle = U | \psi_0 \rangle$.
\end{enumerate}
Suppose we are given a state in a quantum system. By our assumptions, this is a vector in Hilbert space, $\vert \psi \rangle \in \mathcal{H}$. A Hilbert space is a complete vector space that comes equipped with an inner product, the Hermitian form, $\langle\cdot \vert \cdot \rangle : \mathcal{H} \otimes \mathcal{H} \rightarrow \mathbb{R}$ where, for $ X,Y \in \mathcal{H} \; , \; c \in \mathbb{C}$,
\begin{equation*}
\quad \langle X \vert Y \rangle = \overline{\langle Y \vert X \rangle}, \quad \langle cX \vert Y \rangle = \langle X \vert \bar{c}Y \rangle = c \langle X \vert Y \rangle.
\end{equation*}
A Hilbert manifold can be recast as a complexified real space by introducing a complex structure, $J$ such that $J^2 = - \mathds{1}$ 
or in matrix form
\begin{equation*}
J = \begin{bmatrix}
0 & - \mathds{1} \\
\mathds{1} & 0 \\
\end{bmatrix}
\end{equation*}
So that
\begin{equation*}
J^2 = \begin{bmatrix}
0 & - \mathds{1} \\
\mathds{1} & 0 \\
\end{bmatrix} \begin{bmatrix}
0 & - \mathds{1} \\
\mathds{1} & 0 \\
\end{bmatrix}
=
\begin{bmatrix}
- \mathds{1} & 0 \\
0 & - \mathds{1} \\
\end{bmatrix}
= - \mathds{1}.
\end{equation*}
Then the (Hermitean) inner product is
\begin{equation*}
\langle X \vert Y \rangle = g(X,Y) - J\omega(X,Y)
\end{equation*}
where 
\begin{equation}
	g(X,Y) = Re\langle X \vert Y \rangle
\end{equation}
is the \textbf{K\"{a}hler metric} with
\begin{align}\label{eq:kmetric}
g(X,Y) &= g(JX,JY)\\ \nonumber
g(X,Y) &= g(Y,X)
\end{align}
And
\begin{equation}
\omega(X,Y)= Im\langle X \vert Y \rangle
\end{equation} is the symplectic \textbf{K\"{a}hler form} with
\begin{equation*}
\omega(X,Y) = - \omega(Y,X).
\end{equation*}
$g$ and $\omega$ are related by
\begin{equation*}
\omega(X,Y) = g(X,JY) = - g(JX,Y)
\end{equation*}
Introducing local coordinates, $z_i = x_i + i y_i$ , g has the form
\begin{equation}
	g_{i \bar{j}} = \left(\frac{\partial}{\partial x_i}, \frac{\partial}{\partial \bar{y}_j}\right)
\end{equation}
\begin{equation}\label{eq:kahlermetric}\
g = \sum_{i=1}^n dx_i \otimes dx_i + dy_i \otimes dy_i 
\end{equation}
and $\omega$ has the form \cite{article:kahlermetric-kahler}
\begin{align}
	\omega &= i \sum_{i,j} \omega_{i\bar{j}} dz^i \wedge d\bar{z}^j \nonumber\\
		&= \sum_{i=1}^n dx_i \wedge dy_i \label{eq:kform}
\end{align}
Clearly, $g(X,Y) \in {\bigvee}^2$ is symmetric and $\omega(X,Y) \in \bigwedge^2$ is symplectic since $d\omega = 0$.
\\
\\
 J has the coordinate representation
\begin{equation}\label{eq:j}
J = \partial_{X_i} \otimes d Y_i - \partial_{Y_i} \otimes d X_i.
\end{equation}
Summarizing this, the Hilbert space can be expressed as a K\"{a}hler manifold \cite{book:mathtopicbetweencandqm-landsman}:
\begin{defn}
	Let $\mathcal{H}$ be a Hilbert space with a Hermitian product $\langle \cdot \vert \cdot \rangle$ such that, for $X,Y \in \mathcal{H}$ we have $$\langle X \vert Y \rangle = g(X,Y) + J \omega(X,Y)$$ where $J^2 = \mathds{1}$ and $$\omega(X,Y) = g(JX,Y)$$ is the \textit{K\"{a}hler form} and $\omega$ is closed. Then $(\mathcal{H}, \langle \cdot \vert \cdot \rangle, \omega)$ is a \textbf{K\"{a}hler manifold}. We will denote a K\"{a}hler manifold by $\mathcal{K}$.
\end{defn}
\begin{rmk}
	We note that, in local coordinates, the contravariant version of eqs.\ref{eq:kahlermetric},\ref{eq:kform},\ref{eq:j} are given by
	\begin{align}
		Reimann \; tensor \quad& G = \sum_{i=1}^{n}\left(\frac{\partial}{\partial x^i} \otimes \frac{\partial}{\partial x^i} + \frac{\partial}{\partial y^i} \otimes \frac{\partial}{\partial y^i} \right) \label{eq:riemtens}\\
		Symplectic \; form \quad& \Omega = \sum_{i=1}^{n} \left(\frac{\partial}{\partial x^i} \wedge \frac{\partial}{\partial y^i} \right) \label{eq:sympltens}\\
		Complex \; structure \quad& J =\sum_{i=1}^{n} \left(\frac{\partial}{\partial y_i} \otimes dx^i - \frac{\partial}{\partial x^i} \otimes dy_i  \right)\\
		\textnormal{then} \nonumber\\
		G + i\Omega = \sum_{i=1}^{n} \left( \frac{\partial}{\partial q_i} - i \frac{\partial}{\partial q_i} \right) &\otimes \left( \frac{\partial}{\partial q_i} + i \frac{\partial}{\partial q_i} \right) = 4 \sum_{i=1}^{n}\frac{\partial}{\partial z_i} \otimes \frac{\partial}{\partial \bar{z_i}} \label{eq:hilbsum}
	\end{align}
\end{rmk} 
Now, as a physical system, a quantum state is equivalent to a state with a complex constant $\vert \phi \rangle \sim c \vert \psi \rangle$, where $c \in \{\mathbb{C} - 0\}$. Thus, a quantum state is an equivalence class, or a ray, in the projective Hilbert space. The equivalence class will be denoted by $ [X] \in \mathcal{P}_{\mathcal{H}}$ where 
\begin{equation*}
\mathcal{P}_{\mathcal{H}} = \frac{\mathcal{H}}{\sim}
\end{equation*}
is the projective Hilbert space.
\begin{defn}
	Let $ \mathcal{S} = \mathcal{H}_{1} = \{\vert \psi \rangle \in \mathcal{H}, \langle \psi \vert \psi \rangle = 1\} \subset \mathcal{H}$ and $\sim$ defined as above. Then $ \mathcal{P}_{\mathcal{H}} = \mathcal{S}/\sim$ is referred to as the \textbf{quantum phase space} \textnormal{\cite{online:geomformofqm,online:geomformofqm-heydari}}. A physical state is a ray in $\mathcal{P_H}$.
\end{defn}
\noindent
In the event that the states are defined over a standard space we have, $c \vert \psi \rangle \longrightarrow e^{i \theta} \vert \psi \rangle$. Thus, a physical state is a ray in projective Hilbert space where $ \mathcal{S} \simeq S^{2n+1} = \{\psi \in \mathbb{C}^{n+1} \vert \langle \psi \vert \psi \rangle = 1\}$,
\begin{equation*}
	\mathcal{P_H} = \frac{\mathcal{S}}{\sim} \simeq \frac{S^{2n+1}}{U(n)} := \mathbb{C}P(n).
\end{equation*}
That is, a point in the Hilbert space is an equivalence class (a ray), $[\psi ]$, in the projective space $\mathcal{P}_{\mathcal{H}}$ and can be equated to a one-dimensional projection
\begin{equation*}
[\psi] \longleftrightarrow P(\psi) = \vert \psi \rangle \langle \psi \vert.
\end{equation*}
which is a density operator on a pure state. In a normalized space
\begin{equation*}
	[\psi] \longleftrightarrow P(\psi) = \frac{\vert \psi \rangle \langle \psi \vert}{\langle \psi \vert \psi \rangle}.
\end{equation*}
For the finite case this becomes $\mathbb{C}P(n)$. It has a metric given by the Fubini-Study metric \cite{article:geomofqm-kibble},
\begin{align*}
	for \; \psi = &\sum_i \psi_i e_i,\\
		g_{i,j} = &\frac{\langle \psi \vert \psi \rangle \delta_{ij} - \psi_{(i} \bar{\psi}_{j)}}{\langle \psi \vert \psi \rangle^2}
\end{align*}
or in local coordinates
\begin{align}\label{eq:qmmetric}
ds^2 &= g_{i\bar{j}} dz^i d\bar{z}^j \nonumber \\
&= \frac{(1 + \vert \textbf{z} \vert^2) \vert d \textbf{z} \vert^2 - (\bar{\textbf{z}}\cdot d\textbf{z})(\textbf{z}\cdot d\bar{\textbf{z}})}{(1 + \vert \textbf{z} \vert^2)^2} \nonumber \\
&= \frac{(1 + z_i \bar{z}^j)dz_j d\bar{z}^j - \bar{z}^j z_i dz_j d\bar{z}^i}{( 1 + z_i \bar{z}^i)^2} \nonumber\\
\end{align}
\noindent
Let $\widetilde{P} = \langle \psi \vert \phi \rangle$ represent the transition from one state to another. With respect to the Fubini-Study metric, the geodesic distance between the two states is $\gamma(x,y)$. In $\mathbb{C}P^n$ the transition is
\begin{align*}
P_y(x) &= \widetilde{P(x)}_{\pi^{-1}y}(\pi^{-1}(x)) \\
&= \langle \pi^{-1}(y) \vert \pi^{-1}(x) \rangle\\
&= cos^2 \left( \frac{\gamma(x,y)}{\sqrt{2 \hbar}} \right)
\end{align*}
It also has a symplectic structure
\begin{equation*}
	\omega(z,\bar{z}) = \{z,\bar{z}\} \quad z,\bar{z} \in \mathbb{C}P(n)
\end{equation*}
We also have a group, the Unitary group\cite{online:geomofqsys_denandentangl-grabowski-etal}, $U(n)$, that preserves the Hermitian product and is given by the set of complex linear operators, $A \in Gl(n)$ on $\mathcal{H}$ such that $A A^{\dagger} = 1$ where $A^{\dagger}$ is the Hermitian conjugate that satisfies $$\langle AX | Y \rangle = \langle X | A^{\dagger} Y\rangle $$
Its Lie algebra is $\mathfrak{u}(n)$.
\\
\\
The above considerations lead to a bundle representation.
\begin{defn}
	Let $M$ be a manifold. A \textbf{Hilbert bundle}, $(\pi, \mathcal{H}, Q) \; or \; \mathcal{H} \overset{\pi}{\longrightarrow} Q$, is a Banach vector bundle over $Q$ whose fibers consist of a Hilbert space, $\mathcal{H}$ over a manifold, $Q$, with projection $\pi$. The structure group is the unitary group, $U(n)$.
\end{defn}
\begin{rmk}
	In the case where the Hilbert space is replaced with a K\"{a}hler manifold, $\mathcal{K}$, we have the \textbf{K\"{a}hler bundle}, $(\pi, \mathcal{K}, Q) \; or \; \mathcal{K} \overset{\pi}{\longrightarrow} Q$.
\end{rmk}
\noindent
The fiber is given by
\begin{equation*}
\pi^{-1}([\psi]) = \frac{e^{i \psi} \vert \psi \rangle}{\langle \psi \vert \psi \rangle} \in \mathcal{H}_{1}.
\end{equation*}
for a $| \psi \rangle \in \mathcal{H}$ and it has the unitary group, $U(n)$. Thus it forms a principal bundle \cite{online:quantprincfiberbundtopaspect-budzynski,online:geomformofqm-heydari}
\begin{equation*}
U(n) \equiv G \hooklongrightarrow \mathcal{H}_{1} \longrightarrow Q
\end{equation*}
where $G$ is the (Lie) group, $\mathcal{H}_{1}$ is the unit Hilbert space, that is, $\mathcal{H}$ such that $\langle \psi \vert \psi \rangle = 1$ and $Q$ is the base manifold.
The bundle has the adapted coordinates given by \cite{online:geomofqsys_denandentangl-grabowski-etal}
\begin{equation*}
g = dx_i \otimes dx_i + dy_i \otimes dy_i \; , \quad \omega = dx_i \wedge dy_i
\end{equation*}
\noindent
\begin{rmk}
	The tangent bundle takes the form $T\mathcal{K} \simeq \mathcal{K} \times \mathcal{K}$ where $\mathcal{K} = \mathcal{H}_{\mathbb{R}}$, the realification of $\mathcal{H}$ as a Kahler manifold \cite{online:geomofqsys_denandentangl-grabowski-etal}.
\end{rmk}
\begin{expl} \textbf{A qubit}
	\\
	Consider the case of a qubit, a quantum system with two possible states. In this case, $\mathcal{P}_{\mathcal{H}} \simeq \mathcal{H}^{-} / \sim$ where $\mathcal{H}^{-} = \mathcal{H} - \{0\}$ and $x \sim x^{'} \iff x = cx^{'}$ for
	\begin{equation*}
		c \in \mathcal{H}^{-} = \mathcal{H} - \{0\} = S^3 / U(1) \simeq S^2 =: \mathbb{C}P(1)
	\end{equation*}
The states are lines passing through the origin and forms a 2-sphere (called the Bloch sphere). The pure states, $\{ \vert 0 \rangle, \vert 1 \rangle \}$, lie antipodally on the surface of the sphere. The mixed states lie in the interior of the sphere.
	\\
	\\
	Thus, we have a principal fiber bundle with the following structure
	\begin{equation*}
	U(n) \cong S^1 \hookrightarrow S^3 \overset{\eta}{\longrightarrow} \mathbb{C}P^1 \cong S^2 
	\end{equation*}
	\noindent

	$\eta$ is the Hopf fibration \textnormal{\cite{book:geomofminkowspctim-nabor}} and $\mathcal{P_H} \cong \mathbb{C}P^1$.
\end{expl}
\noindent
We now introduce an observable. An observable can be thought of as a constraint as it restricts the number of states available. An observable is defined by a Hamiltonian operator $\hat{H}$.
\begin{defn}
	An \textbf{observable}, $\hat{A}$, is a bounded, self-adjoint, linear operator in Hilbert space.
\end{defn}
\noindent
\begin{rmk}
	The space of observables forms a $C^*$-Algebra. We will use this fact when we consider the Heisenberg representation.\\
\end{rmk}
Suppose, now, we have such an operator, $\hat{A}$, representing an observable in a finite-dimensional Hilbert space such that $\hat{A}^{\dagger} = \hat{A}$. Since the quantum phase space is formed from a principal bundle, the sections form a Lie algebra, $\Gamma (\hat{A}) \subset \mathcal{A}_{\hat{A}}$.
\\
\\
As in the classical case, we want to determine the evolution. We do this by generating a vector space of trajectories, that is, we want the equivalent of the Hamiltonian vector field. We construct a Hamiltonian function that encapsulates the observable in question. The Hamiltonian is represented by a function, $A : \mathcal{P_H} \rightarrow \mathbb{R}$, given by the expectation value
\begin{equation*}
A = \langle \psi \vert \hat{A} \psi \rangle \rightarrow \frac{\langle \psi \vert \hat{A} \psi \rangle}{\langle \psi \vert \psi \rangle}  
\end{equation*}
The function generates a section, $\sigma(A) \in \mathfrak{u}(n)$. As we have noted, the sections form a Lie algebra, $\mathcal{A}_{A}$.
\\
\\
The evolution is obtained from the anchor and is given by
\begin{equation*}
	\rho(\sigma) = \iota_{X_H}\omega = \omega(\sigma, \cdot) \rightarrow X_H = \omega^{-1}(\sigma,\cdot) dH
\end{equation*}
where $X_H$ is determined by the Sch\"{o}dinger equation
\begin{equation*}
	\rho(\sigma_A) = \dot{\psi} = i \hbar \{\psi,H\}
\end{equation*}
and generates vectors in the Hilbert space associated with the observable, $\hat{A}$. A vector has the form,
\begin{equation}
	X_{H_{A}}^i = \omega(X,H_A)^{ij} \frac{\partial}{\partial_j} \quad \textnormal{where}  \quad \omega^{ij} = \omega^{-1}.
\end{equation}
This is a vector subspace, called the \textbf{Schr\"{o}dinger vector field}, of the associated tangent space and is isomorphic to a Hilbert space. This is equivalent to the Hamiltonian vector field of classical mechanics since this takes the sections to a vector subspace of the tangent bundle.
\\
\\
The vector space solution generates a flow, the unitary group of transformations, such that $ \psi^{'}_{t+1} = U_{t} \psi_t$. A value is obtained by calculating the eigenvalues and eigenvectors $A \vert \psi_n \rangle = a_n \vert \psi_n \rangle$, where we have assumed that the $a_n$ are not degenerate.
\\
\\
The extension of the geometrical approach to a quantum system is similar using Lie algebroids. The bundle \cite{book:gmsgeomfromofclassandqm} is given by ($\mathcal{K}$, M, $\pi$) where $\mathcal{K}$ is the K\"{a}hler manifold associated with the projective Hilbert space, M is the underlying base, and $\pi$ is the projection.
\\
\\
\noindent
This yields a eigenvalue equation $A \vert \psi \rangle = a_n \vert \psi \rangle$. A measurement returns an eigenvalue, $a_n$ in a state $\vert \phi \rangle_n$ with a probability given by $p_n = \vert \langle \phi_n \vert \psi \rangle \vert$.
\\
\\
The above constructions can be summarized as
\begin{prop}
	Let $(\mathcal{P_H},\pi,Q,U(n))$ be a quantum phase space that satisfies the following properties:
	\begin{enumerate}
		\item  Let $(\mathcal{P_K},\rho,Q,U(n))$ be a principal fiber bundle where $\mathcal{P_K}$ is a projective K\"{a}hler space with the Hermitian inner product $$\langle \psi \vert \phi \rangle = g(\psi, \phi) + i \omega(\psi, \phi)$$ and $U(n)$ is the Unitary group. This is a quantum phase space.
		\item There is a Hamiltonian, $H$, that generates a Lie algebra of sections, $\Gamma(\mathcal{P_K})$ and forming the observable space, $\mathcal{A} \simeq \Gamma(\mathcal{P_K})$.
		\item  Let $H$ be a Hamiltonian such the sections form a Lie algebra, $\mathcal{A}$, and such that $\rho(H)$ is a bundle morphism from the sections to a vector subspace of the associated tangent bundle, $$\rho(H): \mathcal{A} \rightarrow X_H$$ that generates a flow given by $e^{iHt}$. The anchor map is given by $\rho: \Gamma(H_A) \rightarrow TQ$, where $H_A$ is the Hamiltonian associated with the observable $A$, and maps the sections of $\mathcal{P_K}$ to the associated tangent bundle, $TQ$, that describes the evolution of the system.
		\item The anchor is given by $\rho( \sigma_A) = \dot{\psi} = i \hbar \{ \psi, H_A\}$, which is the Schr\"{o}dinger equation. This results in a Schr\"{o}dinger vector field.
		\item  The Schr\"{o}dinger vector space gives a quantum measurement obtained by the eigenvalue $a_n$ in a state $\vert \phi \rangle_n$ with a probability given by $p_n = \vert \langle \phi_n \vert \psi \rangle \vert$.
		
	\end{enumerate}
Then, the above structure is a Lie algebroid.
\end{prop}
\begin{proof}
	The conditions satisfy the definition of a Lie algebroid as described above.
\end{proof}
\subsection{The Heisenberg Representation}
The mathematical formulation of the Heisenberg representation is given by the dual space of the Lie algebra of the operators via the coadjoint actions of the unitary group \cite{online:geomofqsys_denandentangl-grabowski-etal}, $\mathfrak{u}^*(\mathcal{H})$, of $\mathfrak{u}(\mathcal{H})$. This forms a principal bundle over the unitary group. The initial state remains fixed and the emphasis is on the value of an observable acting on the initial state. The geometry is expressed by considering the Poisson structure of the dual space.
\\
\\
As with the Schr\"{o}dinger representation, we consider pure states in a finite-dimensional model.
\\
\\
We start with a $C^*$-algebra of linear, self-adjoint operators \cite{article:algapproachtoqft-haag-kastler,article:postforqm-segal}. The states are given by positive, linear functionals, $\lambda$, on the $C^*$-algebra, $\mathcal{A}$, and normalized such that
\begin{equation*}
	Tr \lambda = 1
\end{equation*}
The observable space, $\mathcal{O} \subset \mathcal{H}$ is a collection of linear, self-adjoint, that is Hermitian, operators \cite{online:geomofqm-carinena-etal,article:geomformofqm-clemente-gallardo,book:found20centphys-emch,online:geomofqsys_denandentangl-grabowski-etal,book:mathtopicbetweencandqm-landsman}. These observables satisfy the unitary condition, $ A^{\dagger}A = A^{*}A = 1, \; A \in \mathfrak{u}(\mathcal{H}) \subset \mathcal{O}$, associated to $U(n)$, the unitary group, with $n = dim (\mathcal{H})$. This defines a real vector space isomorphic to the Lie algebra $\mathfrak{u}(\mathcal{H})$. Since the geometry is to be found in the dual space we have, for the anti-Hermitian $A \in \mathfrak{u}(\mathcal{H})$ and $\xi \in \mathfrak{u}^*(\mathcal{H})$, the pairing, from the Killing-Cartan form
\begin{equation}\label{eq:operprod}
	\langle \xi , A \rangle = \xi(A) = \frac{i}{2} Tr(A \xi).
\end{equation}  
Furthermore, the elements of $\mathfrak{u}(\mathcal{H})$ have a one-to-one correspondence with the dual Lie algebra, $\mathfrak{u}^*(\mathcal{H})$, given by, for $A \in \mathfrak{u}(\mathcal{H})$ and $\hat{A} \in \mathfrak{u}^*(\mathcal{H})$,  
\begin{align*}
	\mathfrak{u}(\mathcal{H}) &\rightarrow \mathfrak{u}^*(\mathcal{H})\\
	A &\mapsto iA = \hat{A} = \xi
\end{align*}
where we have set $\hat{A} = \xi$. This relates the adjoint and co-adjoint group actions of the group $U(\mathcal{H})$ that is given by $Ad_U(A) = UAU^{\dagger}$. Thus, $\mathfrak{u}(\mathcal{H})$ with Lie bracket $[A,B]_{-} = \frac{1}{i} (AB - BA)$ is a Lie algebra, and has a product given by
\begin{equation}\label{eq:opermet}
	\langle A , B \rangle = \frac{1}{2} Tr(AB).
\end{equation}
(see \cite{online:geomofqsys_denandentangl-grabowski-etal} for more details).
\\
\\
The geometrical structure is given by the orbits of $\mathfrak{u}^*(\mathcal{H})$, the dual space of Hermitian operators on $\mathcal{H}$ \cite{online:geomofqsys_denandentangl-grabowski-etal} with the Poisson product. This results in a principal fiber bundle.
\begin{equation*}
	U(n) \hooklongrightarrow \mathcal{S} \overset{\pi}{\longrightarrow} \mathcal{P_H} \simeq \frac{\mathcal{S}}{\sim} 
\end{equation*}
where the fiber is given by
\begin{equation*}
	\phi = \pi^{-1}([\psi]) = \frac{\langle e^{i \theta} \vert \psi \rangle}{\langle \psi \vert \psi \rangle}
\end{equation*}
\begin{expl}\textbf{A qubit}
	\\
	Consider the qubit again. In this case the principal bundle is given by
	\begin{equation*}
	U(n) \simeq S^1 \hooklongrightarrow \mathcal{S} \simeq S^2 \overset{\pi}{\longrightarrow} \mathcal{P_H} \simeq \frac{\mathcal{S}}{\sim}  \simeq \mathbb{C}P^1
\end{equation*}
here the fiber is given by
\begin{equation*}
\phi = \pi^{-1}([\psi]) = \frac{\langle e^{i \theta} \vert \psi \rangle}{\langle \psi \vert \psi \rangle} \in S^{2n + 1}
\end{equation*}
where we have set $S^{2n + 1} = \{\vert \psi \rangle \in \mathcal{H} \; \vert \; \langle \psi \vert \psi \rangle = 1 \} $.
\end{expl}
Now, consider operators as endomorphisms and let $A,B$ be two endomorphisms in $\mathcal{H}$. We have three products in our K\"{a}hler manifold. The first is the the point product
\begin{equation}
	A \cdot B = AB
\end{equation}
The second is the Jordan form
\begin{equation}
	A \circ B = [A,B]_+ = \frac{1}{2} (AB + BA) \label{eq:jordanprod}
\end{equation}
This generates a \textit{Jordan Algebra} defined by
\begin{defn}
	A \textbf{Jordan Algebra} is a non-associative algebra with a Jordan product, $(\mathcal{O}, \circ)$, such that the product of two elements, $A,B \in \mathcal{O}$, commute, $$A \circ B = B \circ A$$ and $$(AB)A^2 = A(BA^2)$$
\end{defn}
The third product is the \textit{Lie product} given by
\begin{defn}
	Let $A,B \in \mathcal{O}$ and define the \textbf{Lie product} by
	\begin{align*}
		(A,B) &\longmapsto [A,B]_-\\
		[A,B]_- &:= \frac{1}{2}(AB - BA)
	\end{align*}
	This product is closed, anti-symmetric and non-degenerate and defines a symplectic structure.
\end{defn}
In addition, this gives the Poisson form
\begin{equation}
	\{A,B\} = \frac{i}{\hbar}[A,B]_- = \frac{i}{\hbar}(AB - BA)  \label{eq:lieprod}
\end{equation}
So, if we have the functions, $\hat{A}, \hat{B}$ we get\cite{article:geomformofqm-clemente-gallardo}
\begin{equation}
	G(d \hat{A}, d \hat{B}) = \{\hat{A}, \hat{B}\}_+  \label{eq:gprod2}
\end{equation}
and
\begin{equation}
	\Omega(d \hat{A}, d \hat{B}) = \{\hat{A}, \hat{B}\} \label{eq:omegaprod}
\end{equation}
Combining the two products we have
\begin{equation}
	\widehat{AB} = G(d \hat{A}, d \hat{B}) + i \Omega(d \hat{A}, d \hat{B}) = \frac{1}{2}\{\hat{A}, \hat{B}\}_+ + \frac{i}{2}\{\hat{A}, \hat{B}\} \label{eq:hermitprod}
\end{equation}
Finally, we have
\begin{defn}
	A \textbf{Jordan - Lie algebra} is a Lie algebra $(\mathcal{O}, [\cdot, \cdot]_-)$, that possesses a Jordan product satisfying, $\forall A,B,C \in \mathcal{O} \quad and \quad \hbar \in \mathbb{R}$ $$[A, B \circ C] _- = [A, B]_- \circ C + B \circ [A,C]_-$$ and $$(A \circ B ) \circ C - A \circ ( B \circ C) = \frac{1}{4} \hbar [[A,C]_- , B]_- $$
\end{defn}
Using the results obtained above we can form the tensor given by, for $\omega \in \mathfrak{u}^*(\mathcal{H})$,
\begin{equation*}
	R(\omega)(d\hat{A}, d \hat{B}) = \omega([A,B]_+) = \frac{1}{2} Tr \omega(AB + BA) \quad from \; Eq.\ref{eq:jordanprod}
\end{equation*}
Additionally we have a map, for $\hat{A}, \hat{B} \in \mathfrak{u}^*(n) \simeq \mathcal{O}^*$ where the functions $\hat{A}, \hat{B}$ are identified with $A,B \in \mathfrak{u}(n)$ by $\{\hat{A}, \hat{B}\} = \widehat{[A, B]}$. Then $\mathcal{O}^*$ has a Poisson structure and the Poisson bivector, $\Lambda$, has the form
\begin{align}\label{eq:lambda}
	\Lambda(d\hat{A}, d\hat{B})(\omega) &= \{\hat{A} , \hat{B}\}(\omega) \nonumber \\
	&= \omega([A,B]) \nonumber\\
	&= \frac{i}{2} Tr \; \omega(AB - BA)  \quad from \; Eq.\ref{eq:lieprod}
\end{align}
Combining these two results, where, for $\omega \in \mathfrak{u}^*(\mathcal{H})$, we get the tensor
\begin{equation}\label{eq:poissonform}
	(R + i \Lambda)(\omega)(d\hat{A}, d\hat{B}) = 2(\widehat{AB})(\omega) = \omega(AB) = Tr(\omega AB)
\end{equation}
Now, let $\psi$ be a state. Suppose the action, $T$, of $U(\mathcal{H})$ on $\mathcal{H}$ is Hamiltonian. This gives us a momentum (moment) map, $m(\psi)$
\begin{equation*}
	 \left< m(\psi) \Big \vert \frac{1}{i} T \right> = \frac{1}{2}\langle \psi \vert T \psi \rangle_{\mathcal{H}}
\end{equation*}
where, by Eq. \ref{eq:operprod}
\begin{equation}
	\Big \langle m(\psi) \Big \vert \frac{1}{i} T \Big \rangle = \frac{i}{2} Tr \left( m(\psi)\frac{T}{i} \right)=\frac{1}{2} Tr(m(\psi) T) =  \frac{1}{i} Tr \left( m(\psi) \right)
\end{equation}
and this gives us, finally, our momentum map in Dirac notation
\begin{equation*}
	m(\psi) = \frac{\vert \psi \rangle \langle \psi \vert}{\langle \psi \vert \psi \rangle}.
\end{equation*}
The momentum map identifies the tensor given by  Eqs. \ref{eq:gprod2}, \ref{eq:omegaprod}, \ref{eq:poissonform} and the Hermitian product \cite{online:geomofqsys_denandentangl-grabowski-etal} by
\begin{equation*}
	m(G + i \Omega) = R + i \Lambda
\end{equation*}
We can then combine the above results and get
\begin{lemma} \label{lemma:hamvect}
	The momentum map applied to the Jordan-Lie algebra generates a Hamiltonian vector field on $\mathfrak{u}^*(\mathcal{H})$ given by the derviations obtained from Eq. \ref{eq:lambda}. That is,
	we have the derivation
	\begin{equation*}
		\hat{X}_H = \frac{1}{\hbar}\{\hat{H},\cdot\} = \Lambda(d\hat{H},\cdot).
	\end{equation*}
Furthermore, the flow is given by 
	\begin{equation*}
			\frac{d\hat{A}}{dt} =\frac{1}{\hbar} \{\hat{H},\hat{A}\} = \Lambda(d\hat{H},\hat{A}).
	\end{equation*}
\end{lemma}
\begin{proof}
	See \cite{article:geomformofqm-clemente-gallardo} for details.
\end{proof}
	The derivation properties says that the algebra (hence the sections) form a Lie algebra \cite{book:mathtopicbetweencandqm-landsman}. Furthermore, the derivations of the Lie product guaranties that the Leibniz rule is satisfied.
\\
\\
Thus, we have geometrized the Lie algebra, $\mathfrak{u}(\mathcal{H})$ by associating it with a Poisson tensor on the dual vector space, $\mathfrak{u}^*(\mathcal{H})$, forming a principal vector bundle of $U(n)$ on a $\mathbb{C}^*$-algebra. Further details can be found in \cite{online:geomofqm-carinena-etal,article:geomformofqm-clemente-gallardo,online:geomofqsys_denandentangl-grabowski-etal}.
\begin{rmk}
	Every linear operator, $A \in GL(\mathcal{H})$ on $\mathcal{H}$ gives a quadratic function via $f_A(\psi) = H_A = \frac{1}{2} \langle \psi \vert A \psi \rangle$ where $\vert \psi \rangle  \in \mathcal{H}$. If $A$ is Hermitian, $A = A^{\dagger}$, then $f_A(\psi)$ is real. Let $H_A, H_B \in \mathcal{Q}_{\mathcal{K}}$ be two quadratic functions where $\mathcal{Q}_{\mathcal{K}}$ is the set of quadratic functions in the K\"{a}hler manifold $\mathcal{K}$. Then the Riemann tensor (Eq.\ref{eq:gprod2}) and the symplectic tensor (Eq.\ref{eq:omegaprod}) have the form
	\begin{align}
		G(dH_A, dH_B) &= \{H_A, H_B\}_+ = (H_A H_B + H_B H_A) = H_{A\circ B} \label{eq:gprod}\\
		&and \nonumber \\
		\Omega(dH_A, dH_B) &= \{H_A, H_B\} = (H_A H_B - H_B H_A)= H_{[A,B]} \label{eq:quadomega}
	\end{align}
These are combined to form the Hermitian product
	\begin{equation} \label{eq:momprod}
		\langle dH_A \vert dH_B \rangle = G(dH_A, dH_B) + i\Omega(dH_A, dH_B)
	\end{equation}
\begin{rmk}
	Recall that $dH = \omega(H,\cdot)$ and so defines a Hamiltonian vector field, $X_H$. Since we can associate an operator with a function, $\hat{A} \overset{\simeq}{\longrightarrow} A$, then from Eq. \ref{eq:quadomega}, $\Omega$ says that the commutator $[A,B]$ is equivalent to the Poisson bracket of the associated function \cite{article:geommethinqm-spera}.
\end{rmk}
Now, we can define the momentum map, $\mu(\mathcal{H})$, where $\psi \in \mathcal{H}$ and for $\hat{A} \in \mathfrak{u}^*(\mathcal{H})$ associated with $iA \in \mathfrak{u}(\mathcal{H})$, we have (see \cite{article:geomformofqm-clemente-gallardo}, pg. 68)
\begin{equation*}
	\mu^*(\hat{A}) = \frac{\vert \psi \rangle \langle \psi\vert}{\langle \psi \vert \psi \rangle}
\end{equation*}
and, from Eq.\ref{eq:momprod} $\mu_*(G + i \Omega) = R + i \Lambda$ in the contravariant case.
\end{rmk}
Using the results of Lemma \ref{lemma:hamvect}, the evolution is provided by the anchor map which is, in our case, is the Heisenberg equation $$i \hbar \frac{\partial A}{\partial t} = AH - HA$$ where $H$ is the Hamiltonian. It takes the observable sections to the vector space. The anchor generates the Heisenberg vector field and gives the unitary evolution operator $U^{H_A}_{t}$.
\\
\\
The eigenvectors are given by the critical points of the expectation value of the operator, $$A \rightarrow \langle A\rangle(\psi) = \frac{\langle \psi \vert A \psi \rangle}{\langle \psi \vert \psi \rangle}$$ that is $$d\langle A \rangle (\psi_A) = 0 \iff \psi_A \; is \; an \; eigenvector \; of \; A.$$ The result of the measurement is then obtained by calculating the eigenvalues and eigenvectors from the resulting Heisenberg vector field. Thus we are lead to the following proposition:
\begin{prop}
	Let the quantum phase space be given by the $C^*-Algebra$ of self-adjoint operators acting on a Hilbert space, $\mathcal{H}$ and such that
	\begin{enumerate}
		\item The quantum phase space forms a principal vector bundle 
	\begin{equation*}
		U(n) \hooklongrightarrow \mathcal{S} \overset{\pi}{\longrightarrow} \mathcal{P_H} \simeq \frac{\mathcal{S}}{\sim} 
	\end{equation*}
		whose fibers are
		\begin{equation*}
			\phi = \pi^{-1}([\psi]) = \frac{\langle e^{i \theta} \vert \psi \rangle}{\langle \psi \vert \psi \rangle} \in S^{2n + 1}
		\end{equation*} 
		\item The sections are given by the Hamiltonian, $H$, and gives rise to a Lie algebra, $\mathfrak{u}(n)$, that, in turn, generates a Poisson structure on $\mathfrak{u}^*$.
		\item The evolution of the system is given by the anchor, $\rho$, and equal to the Heisenberg equation $$i \hbar \frac{\partial A}{\partial t} = AH - HA$$.
	\end{enumerate}
Then the construction is a Lie algebroid.
\end{prop}
\begin{proof}
	The proof follows from the above discussion. 
	\\
	\\
	The evolution is given by the Heisenberg equation in $\mathfrak{u}(n)$ and has the form $\frac{d}{dt}U = - \frac{i}{\hbar} [H,U]$. Or, in the geometric (operator) form in $\mathfrak{u}^*(n)$, $\frac{d}{dt}\hat{U} = - \frac{i}{\hbar} [H,\hat{U}]$.
	\\
	\\
	The eigenvectors are given by the critical points of the expectation value of the operator, $$A \rightarrow \langle A\rangle(\psi) = \frac{\langle \psi \vert A \psi \rangle}{\langle \psi \vert \psi \rangle}$$ that is $$d\langle A \rangle (\psi_A) = 0 \iff \psi_A \; is \; an \; eigenvector \; of \; A.$$
	\\
	\\
	Thus, we have satisfied the requirements of a Lie algebroid.
\end{proof}
\begin{rmk}
	The Heisenberg uncertainty principle can be obtained from the above work (see \cite{book:geophaseinclassandqm-chrusinski-jamil,online:geomformofqm-heydari} for details). Let $A$ be an observable. Recall from Eq. \ref{eq:gprod} that, for two functions $f_A$ and $f_B$
	\begin{equation*}
		\{f_A, f_B\}_+ = \frac{\hbar}{2} G(X_A, X_B) = \left< \frac{1}{2}[A,B]_+ \right>
	\end{equation*} 
we have that the uncertainty is given by
	\begin{equation*}
		(\Delta A)^2 = \langle A^2 \rangle - \langle A \rangle^2 = \{\tilde{A}, \tilde{A} \} - \tilde{A}^2
	\end{equation*}
where $\tilde{A} = \langle \psi \vert A \vert \psi \rangle$ is the expectation value of the operator $A$ and $\langle \psi \vert \psi \rangle = 1$. Then, if $B$ is another operator and the Schwartz inequality
	\begin{equation*}
		(\Delta A)^2 (\Delta B)^2 \ge \left< \frac{1}{2i}[A, B]\right>^2 + \left< \frac{1}{2} [\delta A, \delta B]_+ \right> ^2
	\end{equation*}
where $\delta A = A - \hat{A}$ and $\delta B = B - \hat{B}$ so that applying Eq.\ref{eq:gprod}, Eq.\ref{eq:quadomega} we end up with
	\begin{equation*}
		(\Delta A)^2 (\Delta B)^2 \ge \Omega(X_A, X_B)^2 + (G(X_A, X_B) - \hat{A} \hat{B})^2
	\end{equation*}
Indeed, it can be shown \cite{article:geommethinqm-spera} that the uncertainty is the result of the curvature.
\\
\\
Thus, the uncertainty is a natural outcome of the geometric process. 
\end{rmk}
\begin{rmk}
	It can also be shown that the geometric representation of quantum mechanics can be used to study entanglement \cite{article:geomaspctsqmandqentangl-chruciski}.
\end{rmk}
\subsection{The Geometrical Equivalence of the Schr\"{o}dinger and Heisenberg Representations}
Finally, we note the equivalence of the Schr\"{o}dinger and Heisenberg representations.
\begin{prop}
	The Heisenberg observable $A^H$ is equivalent to the Schr\"{o}dinger observable, $A^S$, and thus the two representations produce the same results.
\end{prop}
\begin{proof}
	Suppose we have a Schr\"{o}dinger representation state $\psi^S_i$. This can be mapped into a Heisenberg state by $\widetilde{\psi}^{\gamma}_i = U_{\gamma} \psi^{\gamma}_i$ where $U$ is a unitary transformation (usually having the form $U = e^{iHt/\hbar}$) \cite{online:fiberbundleqmiliev}. Then for $\{\widetilde{\psi}_i\}$
	\begin{align*}
	\widetilde{A}^H_{ij} &= \langle \widetilde{\psi}^H | A^S_{ij} \widetilde{\psi}^H_j \rangle\\
	&= \langle U \psi^S_i | A^S U_j \psi^S_j \rangle \\
	&= \langle U^{-1} U_i \psi^S_i | U^{-1}j A^S U_j \psi^S_j \rangle \\
	&= \langle \psi^S_i | U^{-1}_i A^S_{ij} U_j \psi^S_j \rangle\\
	&= U^{-1}_i A^S_{ij} U_j\\
	\end{align*}
Thus, we see that the Heisenberg representation is equivalent to the Schr\"{o}dinger representation.
\end{proof}
\section{Discussion}
The above work demonstrates that Lie algebroids provides a useful context for quantum mechanics. It supports the two main representations and it exposes the underlying structure.
\\
\\
We have not discussed mixed states. For example, in the case of the qubit, the pure states resided on the surface of the Bloch sphere. The mixed state were contained in the interior of the sphere and are given by density matrices \cite{article:geomformofqm-clemente-gallardo}. We have also not considered the environment. That is, we have assumed that the quantum systems are isolated. The environment is usually introduced by way of constraints and causes the states to dissipate or introduces decoherence \cite{online:decoherencemeasprobandinterpqm,online:decohereinselandqmorigofclassic}. Realistically, physical systems are rarely isolated, but we have selected our approach for convenience. 
\\
\\
Also, we have not considered connections, or gauges, in quantum Lie algebroids. This is important in discussing symmetries and Yang-Mills systems. There has been some work on describing connections on Lie algebroids \cite{article:linconninliealgebroids}, but they have not been applied to quantum systems. We have also not explored the Lie algebroid structure of quantum field theory. Quantum fields are nicely described by fiber bundles. This would require extending our considerations to infinite dimensions as well as discussing gauges.
\\
\\
One of the hopes of a geometric theory of quantum mechanics is that it would provide a means of connecting quantum mechanics and general relativity within Lie algebroids. Casting quantum mechanics as a geometric system aids in the process. This remains for future study.
\addcontentsline{toc}{section}{References}
\bibliography{bibtexallqm}
\end{document}